\documentclass[10pt,twocolumn,twoside]{IEEEtran}
\IEEEoverridecommandlockouts

\usepackage{cite}
\usepackage{amsmath,amssymb,amsfonts}
\usepackage{algorithmic}
\usepackage{graphicx}
\usepackage{textcomp}
\usepackage{balance}
\usepackage{xcolor}
\usepackage[utf8]{inputenc}

\usepackage{mathtools}
\usepackage{amsthm} % 用于theorem的字体

\newtheorem{theorem}{Theorem}
\newtheorem{lemma}{Lemma}
\newtheorem{proposition}{Proposition}

\newtheorem{remark}{Remark}
\newtheorem{definition}{Definition}
\newtheorem{assumption}{Assumption}

\usepackage{amssymb}    
\usepackage{mathtools}

\usepackage{algorithm}
\usepackage{algorithmic}

\def\BibTeX{{\rm B\kern-.05em{\sc i\kern-.025em b}\kern-.08em
    T\kern-.1667em\lower.7ex\hbox{E}\kern-.125emX}}
\usepackage{lipsum}

\title{\LARGE \bf 
Incentive-Based Load Curtailment with Limited Information:\\A Bilevel Zeroth-Order Learning Approach
}

\author{Zhisen Jiang, Florian Dörfler, and Saverio Bolognani% <-this % stops a space
\thanks{ 
The authors are with the Automatic Control Laboratory, ETH Zurich, 8092 Zurich, Switzerland. Email: \{zhijiang, dorfler, bsaverio\}@ethz.ch.}%
\thanks{This work was supported by the Swiss Federal Office of Energy (grant SI/502734 MAESTRO) and by the Swiss National Science Foundation under the NCCR Automation (grant agreement 51NF40\textunderscore225155).}
}

\begin{document}
\begingroup
\allowdisplaybreaks

\maketitle

\begin{abstract}
    Incentive-based load curtailment unlocks critical demand-side flexibility but is hindered by the limited knowledge of private user parameters and the inherent nonsmoothness of responses due to physical device constraints. 
    We address this via a constrained bilevel optimization framework and propose the Bi-ZOL (Bilevel Zeroth-Order Learning) algorithm.
    Unlike conventional black-box methods, Bi-ZOL exploits the bilevel structure to decompose the hypergradient, integrating the exact analytical information of the SO’s objective with a zeroth-order estimate of the unknown response sensitivity. 
    This structural decomposition-based learning method mathematically smoothes the nonsmooth response landscape and reduces hypergradient estimation error. 
    We provide theoretical convergence guarantees to an approximate stationary point and demonstrate through simulations that Bi-ZOL achieves near-optimal performance.
\end{abstract}


%%%%%%%%%%%%%%%%%%%%%%%%%%%%%%%%%%%%%%%%%%%%%%%%%%%%%%%%%%%%%%%%%%%%%%%%%%%%%%%%
\section{Introduction}
% [Outline: Background (load control) -> Load control background]
The rapid growth of renewable energy resources and the increasing volatility of demand profiles pose significant challenges for System Operators (SOs) in maintaining power balance. 
Conventionally, when demand exceeds available generation, SOs maintain balance by ramping up or committing reserve generation.
However, relying on peaking units can be economically inefficient and is limited by ramping constraints.
Alternatively, SOs can activate demand-side flexibility by temporarily curtailing load from end-users.
In this context, demand-side flexibility, e.g., controllable load curtailment from end-users, has emerged as a key resource to support system-level power balance.

This paradigm is reflected in real-world power systems where verified demand-side load curtailment is treated as an operational resource.
For instance, the PJM Interconnection has established that demand resources can substitute for generation by providing identifiable capacity services when dispatched~\cite{PJM2025}. 
Similarly, in France, RTE has implemented mechanisms enabling demand-response participants to bid verified load reductions into electricity markets and balancing arrangements~\cite{RTE2018}. 
In California, CAISO has introduced market participation models under which demand-response resources can offer verified load reductions into wholesale markets~\cite{CAISO}.
Furthermore, policy reforms in Japan have explicitly institutionalized "negawatt trading," validating the procurement of demand reductions as a flexibility resource~\cite{Japan2017}.

While these examples demonstrate that demand curtailment can function as an operational substitute for upward generation, the practical question is how the SO reliably and efficiently achieves the desired curtailment.

Broadly, approaches range from direct operator control to incentive-based procurement that leverages voluntary participant responses.
% [Outline: Background (load control) -> Direct Load Control is Impractical]
While Direct Load Control (DLC) has been employed to manage loads, such centralized approaches face privacy concerns. 
They necessitate invasive access to end-users' private constraints and utility functions.
Consequently, there is a paradigm shift towards incentive-based load curtailment.
In this framework, the SO uses financial incentives to induce desired load reduction among self-interested end-users in a privacy-preserving and incentive-compatible manner.
Existing literature has explored various incentive forms.
Differentiated contracts that distinguish between automated and voluntary curtailment are designed based on user opportunity costs  \cite{Daniels2015}.
Risk-aware contract designs have been developed to ensure system reliability while respecting individual utility \cite{Mijatovic2017, Aid2019}.
Recently, dynamic pricing for demand response \cite{Xu2022} and price-based flexibility contracts with competitive guarantees \cite{Zhao2023} are proposed.
Moreover, applications of mechanism design provide the theoretical foundation for ensuring that these incentives remain robust against strategic behavior \cite{Abedrabboh2023}.

From the SO's perspective, the primary objective is to determine the optimal incentives that elicit sufficient response from end-users while minimizing the total payment for the curtailment \cite{Abrishambaf2018}.
However, designing optimal incentives is intrinsically a complex optimization problem characterized by limited knowledge of end-users' private information. 
This opacity necessitates a learning-based approach, where the SO optimizes incentives using observed feedback rather than end-users' private parameters. 
Nevertheless, a critical challenge arises because end-users are subject to strict physical limits (e.g., device capacities). 
When these limits are reached, the response naturally saturates, causing abrupt changes in how load reacts to the incentive. 
This saturation creates ``kinks'' in the response profile, which renders the underlying function mathematically nonsmooth.
This nonsmoothness exposes the limitations of existing learning methods. 
First-order feedback methods \cite{Adam2025}, which rely on response sensitivities estimation, may suffer from oscillation because the gradient of response function is undefined at saturation points. 
Conversely, standard zeroth-order (ZO) methods utilize smoothing schemes to bypass these differentiability issues \cite{maheshwari2024agnostic, Adam2025}.
While these methods are robust to nonsmoothness, they treat the problem as a complete black box. 
They fail to exploit the known structure of the SO's cost function, leading to significant suboptimality and computational inefficiency.

These challenges motivate us to explicitly formulate the incentive-based load curtailment as a constrained bilevel optimization problem, 
a framework that naturally captures the hierarchical structure while explicitly modeling the user-side decision-making process.
To solve this problem, we propose the Bi-ZOL (Bilevel Zeroth-Order Learning) algorithm. 
Distinguishing itself from standard black-box approaches that indiscriminately estimate the gradient of the entire objective function, Bi-ZOL exploits the bilevel structure to decompose the hypergradient via the chain rule. 
This decomposition reveals that the nonsmoothness stems solely from the response sensitivity term (the Jacobian of the lower-level solution). 
Consequently, Bi-ZOL focuses on learning this aggregate response sensitivity directly from end-user feedback via a zeroth-order gradient estimator, which effectively provides a smoothed approximation of the landscape. 
By isolating the learning target, our method incorporates the exact analytical gradient of the SO's known cost function, thereby enhancing optimality while rigorously handling the nonsmooth nature of user behaviors.

In this paper, we first formulate the incentive-based load curtailment as a bilevel optimization problem and derive the piecewise affine structure of the users' response, identifying it as the source of nonsmoothness in Section~\ref{sec:problem_formulation}.
In Section~\ref{sec:solution_algorithm}, we present our Bi-ZOL algorithm and provide theoretical convergence guarantees to an approximate stationary point.
Then, we illustrate the algorithm's performance through numerical simulations in Section~\ref{sec:numerical_experiments}.
Finally, we discuss the remaining open challenges in the Conclusions.

\section{Problem: Incentive-Based Load Curtailment}
\label{sec:problem_formulation}

In this section, we formulate incentive-based load curtailment as a bilevel optimization problem.
The SO, acting as the leader, determines and broadcasts \emph{nodal curtailment incentives} to procure demand reductions from self-interested end-users connected to each node. 
As followers, end-users react to these incentives by optimally adjusting their heterogeneous flexible loads subject to device limits.

Our mathematical formulation is similar to the load control framework in \cite{Adam2025}. 
In contrast to \cite{Adam2025}, which models the end-user response as a generic black-box mapping subject to high-level behavioral properties (e.g., monotonicity), we explicitly characterize the decision-making process as a lower-level optimization problem. 
This structural approach allows us to mathematically derive the specific nature of the nonsmoothness arising from device constraints.
For clarity, we simplify the load control formulation by focusing solely on active load reduction requirements and excluding the underlying grid power flow constraints.

\subsection{SO's upper-level problem (leader)}
\label{subsec:upper_level}
Let $\mathcal{N}$ denote the set of nodes, with $|\mathcal{N}|=N$. 
We define the following vectors in $\mathbb{R}^N$:
\begin{itemize}
    \item[$\boldsymbol{\lambda}$] nodal curtailment incentives, where $\lambda_i$ is the incentive broadcast at node $i$;
    \item[$\mathbf{R}$] aggregate nodal load reduction responses, where $R_i$ is the total reduction at node $i$ contributed by connected end-users;
    \item[$\mathbf{P}^{\mathrm{b}}$] baseline loads, where $P^{\mathrm{b}}_i$ is the baseline load of node $i$.
\end{itemize}

The SO aims to minimize total incentive payments while driving the total system load toward a target $P^{\mathrm{target}}$.
Given a price vector $\boldsymbol{\lambda}$ and response $\mathbf{R}$,
we define the global load mismatch as:
\begin{equation}
    E \coloneqq \mathbf{1}^\top\!\big(\mathbf{P}^{\mathrm{b}}-\mathbf{R}\big)\;-\;P^{\mathrm{target}}.
\label{eq:E}
\end{equation}
The SO solves the following optimization problem:
\begin{equation}
\min_{\boldsymbol{\lambda}\in \Lambda, \mathbf{R}} \;\varphi (\boldsymbol{\lambda,\mathbf{R}} )\;\coloneqq \;\boldsymbol{\lambda}^{\top}\mathbf{R}\;+\;\rho E^2,
\label{eq:upper_problem}
\end{equation}
where $\rho>0$ is a penalty parameter and $\Lambda$ is the admissible price set.

Note that $\varphi(\boldsymbol{\lambda}, \mathbf{R})$ is smooth and Lipschitz continuous on compact sets.
We adopt the following standard assumption regarding the constraints:
\begin{assumption} \label{ass:bounded_lambda}
    The set $\Lambda$ is nonempty, convex, compact, and bounded with diameter $D$, i.e., $\left\| \boldsymbol{\lambda}_1-\boldsymbol{\lambda}_2 \right\| \le D$ for any $\boldsymbol{\lambda}_1, \boldsymbol{\lambda}_2 \in \Lambda$.  
\end{assumption}
This assumption is readily satisfied in practice, typically taking the form of box constraints $\Lambda \coloneqq [\underline{\boldsymbol{\lambda}}, \overline{\boldsymbol{\lambda}}]$.

\subsection{End-users' lower-level problem (followers)}
\label{subsec:lower_level}
End-users are partitioned by nodes. 
Let $\mathcal{U}$ be the set of users and $\mathcal{U}_i\subseteq\mathcal{U}$ denote the subset of users connected to node $i$ (where $\{\mathcal{U}_i\}_{i\in\mathcal{N}}$ are disjoint).
Each user $u$ controls a set of flexible devices indexed by $\mathcal{K}_u$.
For each device $k\in\mathcal{K}_u$, let $r_{u,k}\in\mathbb{R}$ denote the load reduction, subject to the capacity constraint:
\begin{equation}
0\le r_{u,k}\le C_{u,k},
\label{eq:device_capacity}
\end{equation}
where $C_{u,k}>0$ represents the maximum load reduction for the device, which is determined by the device capacity.

Given the nodal incentive $\lambda_i$, each end-user $u\in\mathcal{U}_i$ minimizes their net cost (discomfort minus incentive payment).
The decision problem for user $u$ is formulated as:
\begin{subequations}\label{eq:user_lower}
\begin{align}
\min_{\{r_{u,k}\}_{k\in\mathcal{K}_u}}\;\;&
\sum_{k\in\mathcal{K}_u}\Big(\frac{1}{2\alpha_{u,k}}\,r_{u,k}^2-\lambda_i r_{u,k}\Big)
\label{eq:user_lower_obj}
\\
\text{s.t.}\;\;&
0\le r_{u,k}\le C_{u,k},\qquad \forall k\in\mathcal{K}_u,
\label{eq:user_lower_con}
\end{align}
\end{subequations}
where $\alpha_{u,k}>0$ denotes the device discomfort sensitivity.
Unlike \cite{Adam2025}, our formulation explicitly models the aggregation of heterogeneous devices.
The quadratic term captures the \emph{increasing marginal disutility} of load reduction. For example, while dimming lights slightly may cause negligible annoyance, curtailing HVAC usage during extreme weather incurs significantly higher discomfort per unit of energy reduced.

\begin{remark}
The proposed quadratic discomfort model \eqref{eq:user_lower_obj} implies that all devices begin reducing load as soon as $\lambda_i > 0$. 
This formulation can be trivially extended to include a linear discomfort term $\beta_{u,k} r_{u,k}$ (where $\beta_{u,k} > 0$), representing a minimum "activation cost" for device $k$.
Physically, a device only participates when the incentive $\lambda_i$ exceeds its specific discomfort threshold $\beta_{u,k}$.
This allows the model to capture sequential activation behaviors, where users prioritize ``cheaper'' reduction before committing to ``expensive'' devices as incentives increase.
\end{remark}

Let $\{r_{u,k}^\star(\lambda_i)\}_{k\in\mathcal{K}_u}$ denote the optimal solution to \eqref{eq:user_lower}.
The aggregate nodal response is then given by:
\begin{equation}
R_i(\lambda_i) \coloneqq \sum_{u\in\mathcal{U}_i}\sum_{k\in\mathcal{K}_u} r_{u,k}^\star(\lambda_i).
\label{eq:nodal_aggregate_response}
\end{equation}
We now characterize the properties of $\mathbf{R}$ to identify the source of nonsmoothness in the bilevel problem.

\begin{proposition}
\label{prop:response_function_property}
Fix a node $i\in\mathcal{N}$. The aggregate nodal response $R_i(\lambda_i)$ satisfies the following properties:
\begin{enumerate}
\item (\emph{Uniqueness}) 
 For every $\lambda_i\in\mathbb{R}$, the response $R_i(\lambda_i)$ is unique.
\item (\emph{Closed-form piecewise affine structure}) 
$R_i(\lambda_i)$ admits the closed-form expression:
\begin{equation}
R_i(\lambda_i)=\sum_{u\in\mathcal{U}_i}\sum_{k\in\mathcal{K}_u}\min\big\{ C_{u,k}, \max\{0, \alpha_{u,k}\lambda_i\} \big\}.
\label{eq:R_i_closed_form}
\end{equation}
Consequently, $R_i(\cdot)$ is piecewise affine with breakpoints contained in the set $\{0\}\cup\{C_{u,k}/\alpha_{u,k}:k\in\mathcal{K}_u, u\in\mathcal{U}_i\}$.
\item (\emph{Lipschitz Continuity}) 
$R_i(\cdot)$ is globally Lipschitz continuous with constant:
\begin{equation*}
    L_i \coloneqq \sum_{u\in\mathcal{U}_i}\sum_{k\in\mathcal{K}_u}\alpha_{u,k}.
\end{equation*}
\end{enumerate}
\end{proposition}

\begin{IEEEproof}
    We start the proof with the device optimal solution $r_{u,k}^\star(\lambda_i)$ admitted by \eqref{eq:user_lower}.
    Since $\alpha_{u,k}>0$ for all $k$, the Hessian matrix is diagonal with strictly positive entries $\{\frac{1}{\alpha_{u,k}}\}_{k\in\mathcal{K}_u}$, the objective is strictly convex in $\{r_{u,k}\}$.
    Together with the nonempty compact constraints \eqref{eq:user_lower_con} and first order optimality condition, for any device $(u,k)$ and any $\lambda_i\in\mathbb{R}$, the problem \eqref{eq:user_lower} admits a unique optimizer given by
\begin{equation}
r_{u,k}^\star(\lambda_i)
=
\min\big\{ C_{u,k}, \max\{0, \alpha_{u,k}\lambda_i\} \big\}.
\label{eq:clip_closed_form}
\end{equation}
Summing these unique device responses over all users at node $i$ proves statements (1) and (2).
Regarding (3), $r_{u,k}^\star(\cdot)$ has slope either $0$ or $\alpha_{u,k}$ (except at two breakpoints),
so it is globally Lipschitz with constant $\alpha_{u,k}$. 
Therefore,
\begin{align*}
    |R_i(\lambda )-R_i(\lambda' )| & \le \sum_{u\in\mathcal{U}_i}\sum_{k\in \mathcal{K} _u}{|r_{u,k}^{*}(\lambda )}-r_{u,k}^{*}(\lambda')|
\\
& \le \sum_{u\in\mathcal{U}_i}\sum_{k\in \mathcal{K} _u}{\alpha _{u,k}|\lambda}-\lambda' |
\\
& =L_i|\lambda -\lambda' |.
\end{align*}

\end{IEEEproof}

To illustrate the piecewise affine responses in Proposition~\ref{prop:response_function_property}, 
Fig.~\ref{fig:response_schematic} depicts responses for a node with three connected end-users. 
The individual device responses exhibit saturation at their capacities $C_1$, $C_2$, $C_3$, 
and the aggregate $R_i$ is the sum with breakpoints at $\{C_{u,k}/\alpha_{u,k}\}$ as in~\eqref{eq:R_i_closed_form}.

\begin{figure}[htb]
    \centering
    \includegraphics[width=\columnwidth]{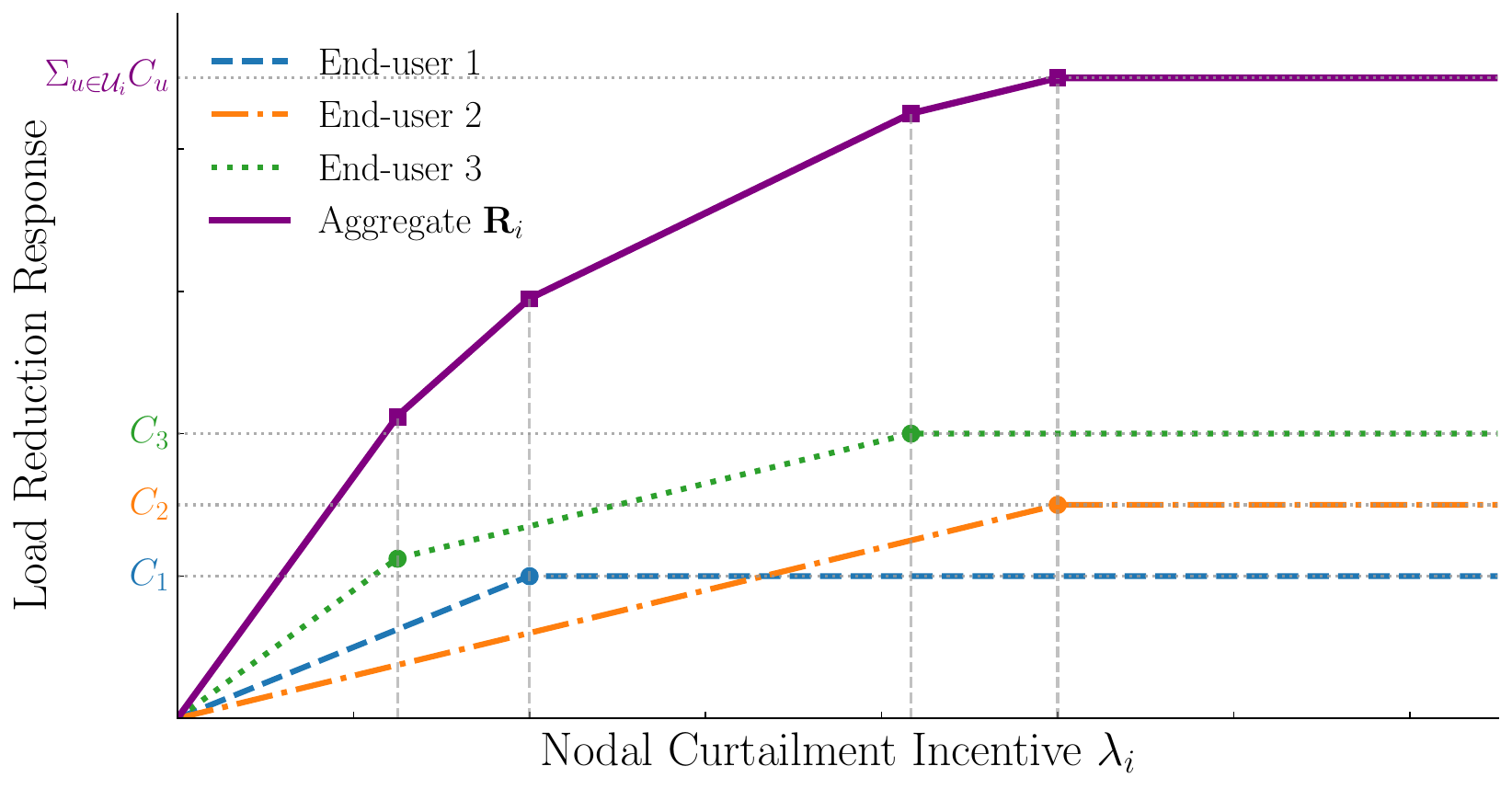}
    \caption{Responses for a node with three connected end-users: individual load reduction responses and aggregate nodal response $R_i$ as functions of the nodal incentive $\lambda_i$.}
    \label{fig:response_schematic}
\end{figure}

\begin{lemma}
    \label{lemma:R_vector_Lipschitz}
    The aggregate nodal response vector $\mathbf{R}(\boldsymbol{\lambda})$ is globally Lipschitz continuous with constant $L_R \coloneqq \max_{i \in \mathcal{N}} L_i$, and is bounded by $R_{\max} \coloneqq \sqrt{\sum_{i=1}^N (\sum_{u\in\mathcal{U}_i}\sum_{k\in\mathcal{K}_u} C_{u,k})^2}$.
\end{lemma}
\begin{proof}
    Using the vector components:
    \begin{align*}
        \|\mathbf{R}(\boldsymbol{\lambda})-\mathbf{R}(\boldsymbol{\lambda}')\|
         & = \sqrt{\sum_{i=1}^N |R_i(\lambda_i)-R_i(\lambda_i')|^2}
        \\
        & \le \sqrt{\sum_{i=1}^N L_i^2 |\lambda_i-\lambda_i'|^2}
        \\
        & \le \sqrt{\left(\max_{j} L_j^2\right) \sum_{i=1}^N |\lambda_i-\lambda_i'|^2}
        \\
        & = L_R \|\boldsymbol{\lambda}-\boldsymbol{\lambda}'\|_2.
    \end{align*}
    Boundedness follows directly from the fact that $0 \le R_i \le \sum_{u,k} C_{u,k}$.
\end{proof}

\subsection{Bilevel Formulation}
\label{subsec:bilevel}

We now combine the SO’s upper-level problem \eqref{eq:upper_problem} and the end-users’ lower-level problem \eqref{eq:user_lower}.
The resulting leader--follower interaction can be written as the following bilevel program:
\begin{subequations}
\label{eq:full_bilevel}
\begin{align}
\min_{\boldsymbol{\lambda}\in\Lambda, \mathbf{R}}\;\; &
\varphi (\boldsymbol{\lambda},\mathbf{R})\;\coloneqq \;\boldsymbol{\lambda}^{\top}\mathbf{R}\;+\;\rho E^2
\label{eq:bilevel_upper_obj}
\\
\text{s.t.}\;\;
&
R_i = \sum_{u\in\mathcal{U}_i}\sum_{k\in\mathcal{K}_u} r_{u,k}^\star(\lambda_i), \quad \forall i\in\mathcal{N}, \notag\\
&
E = \mathbf{1}^\top\!\big(\mathbf{P}^{\mathrm{b}}-\mathbf{R}\big)\;-\; P^{\mathrm{target}},\notag\\
&
\forall u\in\mathcal{U}: \notag\\
&
r_{u,k}^\star(\lambda_i)=\arg\min_{k\in\mathcal{K}_u}
\left\{
    \sum_{k\in\mathcal{K}_u}\Big(\frac{1}{2\alpha_{u,k}}\,r_{u,k}^2-\lambda_i r_{u,k}\Big)
\right\} \label{eq:equilibrium_constraint_obj}\\
&
\hspace{2cm} \text{s.t.}\ 0\le r_{u,k}\le C_{u,k}, \quad \forall k\in\mathcal{K}_u.
\label{eq:equilibrium_constraint_con}
\end{align}
\end{subequations}

The bilevel formulation \eqref{eq:full_bilevel} captures the SO’s incentive-based load curtailment under limited information.
The SO can observe the aggregate nodal response $\mathbf{R}$ after broadcasting the price vector $\boldsymbol{\lambda}$.
However, it does not know the closed-form mapping from the incentive vector to the aggregate response because it is induced by private end-user discomfort parameters and device constraints.

The general interaction between the SO and the end-users is illustrated in Fig.~\ref{fig:bilevel_interaction}.

\begin{figure}[htb]
    \centering
    \includegraphics[width=\columnwidth]{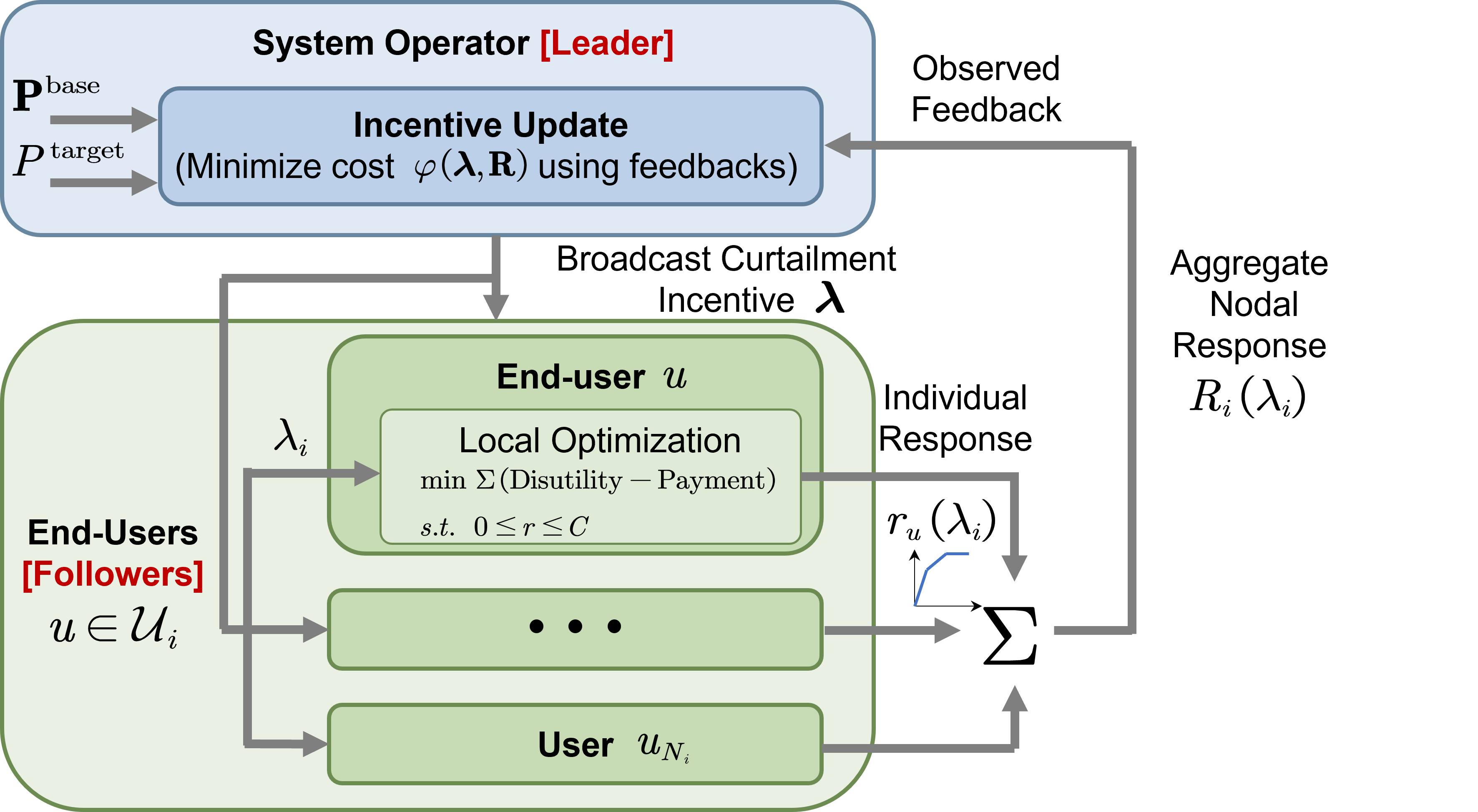}
    \caption{Incentive-based load curtailment architecture.}
    \label{fig:bilevel_interaction}
\end{figure}

%%%%%%%%%%%%%%%%%%%%%%%%%%%%%%%%%%%%%%%%%%%%%%%%%%%%%%%%%%%%%%%%%%%%%%%%%%%%%%%%
\section{Solution: Bilevel Zeroth-Order Learning}
\label{sec:solution_algorithm}
In this section, we propose the Bilevel Zeroth-Order Learning (Bi-ZOL) algorithm to solve the bilevel problem \eqref{eq:full_bilevel}.
The core principle is to learn the sensitivity of the aggregate nodal response to incentives from observed feedback using a zeroth-order estimator.
From a practical perspective, this learning process circumvents the SO's limited knowledge regarding end-user parameters.
Mathematically, this corresponds to a randomized smoothing of the nonsmooth response landscape, which is essential for ensuring algorithmic convergence.
Furthermore, unlike standard black-box methods, Bi-ZOL fully utilizes the SO's available analytical information of the upper-level cost function to enhance the precision of the hypergradient estimate.

\subsection{Approximate Stationary Point}
\label{subsec:approximate_stationary_point}
To formally introduce the algorithm, we first define the stationary point condition we aim to satisfy.
We view the bilevel problem \eqref{eq:full_bilevel} as a general constrained nonsmooth nonconvex optimization problem:
\begin{equation}
    \label{eq:nonsmooth_optimization}
    \min_{\boldsymbol{\lambda}\in\Lambda}\;\; \tilde{\varphi}\left( \boldsymbol{\lambda } \right) \coloneqq \varphi (\boldsymbol{\lambda },\mathbf{R}(\boldsymbol{\lambda}))
    =\boldsymbol{\lambda}^{\top}\mathbf{R}(\boldsymbol{\lambda})\;+\;\rho E\left( \boldsymbol{\lambda } \right) ^2,  
\end{equation}
where the equilibrium constraints \eqref{eq:equilibrium_constraint_obj}-\eqref{eq:equilibrium_constraint_con} are implicitly captured by the single-valued mapping $\mathbf{R}(\boldsymbol{\lambda})$.

Wherever $\tilde{\varphi}(\cdot)$ is differentiable, the hypergradient is given by the chain rule:
\begin{equation}
    \label{eq:hypergradient}
    \nabla \tilde{\varphi}( \boldsymbol{\lambda } ) = \nabla_1 \varphi( \boldsymbol{\lambda },\mathbf{R} ) + \mathbf{J} \mathbf{R}( \boldsymbol{\lambda} )^\top \nabla_2 \varphi( \boldsymbol{\lambda },\mathbf{R} ),
\end{equation}
where the Jacobian matrix $\mathbf{J} \mathbf{R}( \boldsymbol{\lambda} )$ captures the sensitivity of the aggregate response to incentives.
However, as discussed in Section \ref{subsec:lower_level}, $\mathbf{R}(\boldsymbol{\lambda})$ is nonsmooth, possessing numerous points where the Jacobian is undefined. 
This renders standard gradient-based methods invalid.

To address this, in our algorithm, we propose to approximate the ill-defined hypergradient by replacing the $\mathbf{J} \mathbf{R}( \boldsymbol{\lambda} )$ with the approximate Jacobian $\mathbf{J}\mathbf{R}_{\delta}(\boldsymbol{\lambda})$, where $\mathbf{R}_{\delta}(\boldsymbol{\lambda})$ is the smoothed response function by randomized smoothing defined as:
\begin{equation}
    \label{eq:smoothed_response}
    \mathbf{R}_{\delta}(\boldsymbol{\lambda}) = \mathbb{E}_{\mathbf{u}\sim \mathrm{Unif}(\mathbb{B})} \left[ \mathbf{R}(\boldsymbol{\lambda} + \delta \mathbf{u}) \right],
\end{equation}
where $\mathrm{Unif}(\mathbb{B})$ is the uniform distribution over the closed unit ball and $\delta > 0$ is the smoothing radius.
Consequently, we have our approximate hypergradient:
\begin{equation}
    \label{eq:approximate_hypergradient}
    \nabla \tilde{\varphi}_{\delta}( \boldsymbol{\lambda } ) = \nabla_1 \varphi( \boldsymbol{\lambda },\mathbf{R} ) + \mathbf{J}\mathbf{R}_{\delta}(\boldsymbol{\lambda})^\top \nabla_2 \varphi( \boldsymbol{\lambda },\mathbf{R} ).
\end{equation}
Crucially, unlike fully black-box approaches, our formulation retains the exact values of the partial derivatives $\nabla_1 \varphi$ and $\nabla_2 \varphi$.
By isolating the zeroth-order estimation solely to the unknown response sensitivity, \eqref{eq:approximate_hypergradient} significantly reduces estimation error and represents the true descent direction more faithfully.

Note that $\nabla \tilde{\varphi}_{\delta}$ is defined everywhere, and it serves as a proxy for the ill-defined hypergradient. 
Since the true objective violates the standard Lipschitz gradient assumption, we cannot rely on the smooth case first-order optimality conditions ($\nabla\tilde{\varphi}(\boldsymbol{\lambda}^\star)=0$).
Moreover, \eqref{eq:nonsmooth_optimization} is constrained.
Thus, we define the $(\delta, \epsilon)$-Frank-Wolfe Stationary Point (FWSP) to characterize stationarity for the constrained, nonsmooth landscape.
\begin{definition}
    \label{def:CFWSP}
    A point $\boldsymbol{\lambda} \in \Lambda$ is a $(\delta, \epsilon)$-Frank-Wolfe Stationary Point of problem \eqref{eq:nonsmooth_optimization} if:
\begin{equation}
    \underset{\mathbf{z}\in \Lambda}{\max}\left< \mathbf{z}-\boldsymbol{\lambda },-\nabla \tilde{\varphi}_{\delta}(\boldsymbol{\lambda }) \right> \le \epsilon.
\end{equation}
\end{definition}
This definition relaxes the standard stationarity condition in nonsmooth optimization by allowing for a smoothing error $\delta$ and a convergence tolerance $\epsilon$, which is consistent with the literature on zeroth-order optimization for nonsmooth functions \cite{Liu2024zeroth}.

\begin{remark}
The rationale for this stationarity definition is grounded in nonsmooth analysis. 
For a nonsmooth function, standard gradients are ill-defined.
Instead, optimality is typically characterized via the Clarke subdifferential $\partial^C \tilde{\varphi}(\boldsymbol{\lambda})$.
Our definition is inspired by the concept of the Goldstein subdifferential, a computationally tractable relaxation of $\partial^C \tilde{\varphi}(\boldsymbol{\lambda})$ constructed by sampling gradients within a $\delta$-neighborhood.
Our $\nabla \tilde{\varphi}_{\delta}$ serves as a reliable proxy for the Goldstein subdifferential.
As the smoothing radius $\delta \to 0$, any limit point of the sequence of $(\delta, \epsilon)$-FWSPs satisfies the necessary optimality condition defined by $\partial^C \tilde{\varphi}(\boldsymbol{\lambda})$.
\end{remark}

\subsection{Bilevel Zeroth-Order Learning Algorithm}
\label{subsec:algorithm}

The proposed Bi-ZOL algorithm employs a Frank-Wolfe (FW) update scheme and operates by alternating between learning the response sensitivity via a two-point estimator and updating the incentives along a descent direction.

We leverage the fact that the SO has full analytical knowledge of the upper-level objective $\varphi$.
At iteration $k$, given the current pair $(\boldsymbol{\lambda}_k, \mathbf{R}_k)$, the partial gradients $\nabla_1\varphi$ and $\nabla_2\varphi$ are computed exactly:
\begin{equation}
    \label{eq:partials_closedform}
    \nabla_1\varphi(\boldsymbol{\lambda},\mathbf{R})=\mathbf{R}_k,
    \quad
    \nabla_2\varphi(\boldsymbol{\lambda},\mathbf{R})=\boldsymbol{\lambda}_k-2\rho E(\mathbf{R}_k)\mathbf{1}.
\end{equation}
The remaining challenge is to estimate the smoothed Jacobian $\mathbf{J}\mathbf{R}_{\delta}(\boldsymbol{\lambda}_k)$.
It is a known result in zeroth-order optimization that the gradient of a function smoothed over the unit ball can be estimated using samples from the unit sphere.
Accordingly, we generate a random direction $\mathbf{w}\sim \mathrm{Unif}(\mathbb{S}^{N-1})$ and utilize a two-point estimator:
\begin{equation}
    \label{eq:jac_estimator}
    \widehat{\mathbf{J}}{\mathbf{R}}(\boldsymbol{\lambda}_k)
    \;=\;
    \frac{N}{2\delta}
    \Big(\mathbf{R}(\boldsymbol{\lambda}_k+\delta \mathbf{w})-\mathbf{R}(\boldsymbol{\lambda}_k-\delta \mathbf{w})\Big)\mathbf{w}^\top.
\end{equation}
This estimator serves two critical functions: 
(1) It allows the SO to learn the aggregate behavior of end-users without accessing their private parameters; 
(2) It effectively estimates the Jacobian of the \emph{smoothed} response $\mathbf{R}_\delta$, thereby addressing the nonsmoothness issues of the original function $\mathbf{R}$.

Combining \eqref{eq:partials_closedform} and \eqref{eq:jac_estimator}, Bi-ZOL constructs the hypergradient estimate:
\begin{equation}
\label{eq:hypergrad_est}
\widehat{\mathbf{g}}_k
\;=\;
\nabla_1\varphi(\boldsymbol{\lambda}_k,\mathbf{R}_k)
+
\widehat{\mathbf{J}}{\mathbf{R}}(\boldsymbol{\lambda}_k)^\top
\nabla_2\varphi(\boldsymbol{\lambda}_k,\mathbf{R}_k).
\end{equation}

With $\widehat{\mathbf{g}}_k$, the Frank-Wolfe linear oracle computes the descent direction:
\begin{equation}
\label{eq:fw_oracle}
\mathbf{z}_k\in\arg\max_{\mathbf{z}\in\Lambda}\ \langle \mathbf{z},-\widehat{\mathbf{g}}_k\rangle,
\end{equation}
followed by the update step:
\begin{equation}
\label{eq:fw_update}
\boldsymbol{\lambda}_{k+1}=\boldsymbol{\lambda}_k+\gamma\big(\mathbf{z}_k-\boldsymbol{\lambda}_k\big),
\end{equation}
where $\gamma\in(0,1]$ is the stepsize. 
Since $\Lambda$ is convex, the update \eqref{eq:fw_update} ensures that the new incentives remain feasible without requiring projection.

The complete procedure is summarized in Algorithm \ref{alg:bizol}.

\begin{remark}
In a practical implementation, the iterative updates of $\boldsymbol{\lambda}$ and the probing of $\mathbf{R}(\boldsymbol{\lambda} \pm \delta \mathbf{w})$ occur within a digital negotiation phase (e.g., via Home Energy Management Systems) prior to the physical load adjustment.
The SO exchanges incentives with the HEMS agents to estimate sensitivities without physically altering the load consumption.
This avoids the latency and frequent physical actuation.
Moreover, it is possible to adapt the algorithm to use a one-point estimator \cite{Zhang2022OnePointEsti,Zhiyu_modelfree}. 
A one-point approach would require only a single broadcast per iteration, simplifying the communication overhead. We leave the exploration of one-point variants for real-time operations to future work.
\end{remark}

\begin{algorithm}[t]
\caption{Bilevel Zeroth-Order Learning (Bi-ZOL)}
\label{alg:bizol}
\begin{algorithmic}[1]
\STATE \textbf{Input:} Stepsize $\gamma$; smoothing radius $\delta>0$.
\STATE \textbf{Initialize:} Choose $\boldsymbol{\lambda}_0\in\Lambda$.
\FOR{$k=0,1,2,\ldots$}
    \STATE \textbf{Query response:} Broadcast $\boldsymbol{\lambda}_k$, observe $\mathbf{R}_k=\mathbf{R}(\boldsymbol{\lambda}_k)$.
    \STATE \textbf{Compute mismatch:} $E_k=\mathbf{1}^\top\!\big(\mathbf{P}^{\mathrm{b}}-\mathbf{R}_k\big)-P^{\mathrm{target}}$.
        \STATE \textbf{Compute partials:} $\nabla_1\varphi=\mathbf{R}_k$; $\nabla_2\varphi=\boldsymbol{\lambda}_k-2\rho E_k\mathbf{1}$.
    \STATE \textbf{ZO Jacobian estimation:} 
    \STATE \quad Sample direction $\mathbf{w}\sim \mathrm{Unif}(\mathbb{S}^{N-1})$.
    \STATE \quad Broadcast $\boldsymbol{\lambda}_k+\delta\mathbf{w}$, observe $\mathbf{R}^{+}=\mathbf{R}(\boldsymbol{\lambda}_k+\delta\mathbf{w})$.
    \STATE \quad Broadcast $\boldsymbol{\lambda}_k-\delta\mathbf{w}$, observe $\mathbf{R}^{-}=\mathbf{R}(\boldsymbol{\lambda}_k-\delta\mathbf{w})$.
    \STATE \quad Estimate $\widehat{\mathbf{J}}{\mathbf{R}}(\boldsymbol{\lambda}_k)\leftarrow \dfrac{N}{2\delta}\big(\mathbf{R}^{+}-\mathbf{R}^{-}\big)\mathbf{w}^\top$.
    \STATE \textbf{Hypergradient estimate:} 
    \STATE \quad $\widehat{\mathbf{g}}_k\leftarrow \nabla_1\varphi + \widehat{\mathbf{J}}{\mathbf{R}}(\boldsymbol{\lambda}_k)^\top \nabla_2\varphi$.
    \STATE \textbf{FW linear oracle:} $\mathbf{z}_k\in\arg\max_{\mathbf{z}\in\Lambda}\ \langle \mathbf{z},-\widehat{\mathbf{g}}_k\rangle$.
    \STATE \textbf{FW update:} $\boldsymbol{\lambda}_{k+1}\leftarrow \boldsymbol{\lambda}_k+\gamma\big(\mathbf{z}_k-\boldsymbol{\lambda}_k\big)$.
\ENDFOR
\end{algorithmic}
\end{algorithm}

\subsection{Algorithm Convergence}
\label{subsec:algorithm_convergence}
In this section, we establish the convergence of the Bi-ZOL algorithm to a $(\delta, \epsilon)$-FWSP. 
The analysis relies on the smoothness properties of the smoothed response function $\mathbf{R}_{\delta}$ and the ancillary function $\bar{\varphi}_{\delta}(\boldsymbol{\lambda}) \coloneqq \varphi(\boldsymbol{\lambda}, \mathbf{R}_{\delta}(\boldsymbol{\lambda}))$.
By quantifying the bias introduced by smoothing and the variance of the zeroth-order estimator, we derive the following convergence guarantee.

\begin{theorem}
    \label{thm:convergence}
    With Assumption~\ref{ass:bounded_lambda}, the Bi-ZOL algorithm (Algorithm~\ref{alg:bizol}) with constant stepsize $\gamma \in (0,1]$ guarantees the following bound on the expected Frank-Wolfe gap after $T$ iterations:
    \begin{equation}
    \label{eq:convergence_bound}
    \frac{1}{T}\sum_{k=0}^{T-1} \mathbb{E} \left[ \mathcal{G}_{\delta}(\boldsymbol{\lambda}_k) \right]
    \le
    \frac{\Delta_0}{\gamma T} 
    + 
    \frac{L_{\nabla \bar{\varphi}_{\delta}} D^2}{2} \gamma
    + 
    C_{\mathrm{bias}} D \delta
    + 
    C_{\mathrm{noise}} D,
    \end{equation}
    where $\mathcal{G}_{\delta}(\boldsymbol{\lambda}) \coloneqq \max_{\mathbf{z}\in\Lambda} \langle \mathbf{z}-\boldsymbol{\lambda}, -\nabla \tilde{\varphi}_{\delta}(\boldsymbol{\lambda}) \rangle$ is the FW gap for the approximate problem.
    The constants are defined as:
    \begin{itemize}
        \item $\Delta_0 \coloneqq \bar{\varphi}_{\delta}(\boldsymbol{\lambda}_0) - \min_{\boldsymbol{\lambda}}\bar{\varphi}_{\delta}(\boldsymbol{\lambda})$,
        \item $L_{\nabla \bar{\varphi}_{\delta}} \coloneqq L_R(2+2\rho N L_R) + \frac{c N L_R}{\delta}(\Lambda_{\max} + 2\rho\sqrt{N}(|c_0|+\sqrt{N}R_{\max}))$,
        \item $C_{\mathrm{bias}} \coloneqq (1+2\rho N L_R)L_R$,
        \item $C_{\mathrm{noise}} \coloneqq 4 M_{\varphi 2} (2\pi)^{1/4} N L_R$, with $M_{\varphi 2} \coloneqq \Lambda_{\max} + 2\rho\sqrt{N}(|c_0|+\sqrt{N}R_{\max})$.
    \end{itemize}
\end{theorem}

The proof of Theorem \ref{thm:convergence}, along with the supporting lemmas regarding the properties of the estimator and the ancillary function, are provided in Appendix.

\begin{remark}
Note that we use the average of the expected Frank-Wolfe gap as the convergence measure, as considered in \cite{Zhiyu_modelfree}. 
This is because it equals to $\mathbb{E} \left[ \max_{\mathbf{z}\in \Lambda}\left< -\nabla \tilde{\varphi}_{\delta}\left( \boldsymbol{\lambda}_R \right) ,\mathbf{z}-\boldsymbol{\lambda}_R \right> \right]$ where $\boldsymbol{\lambda}_R$ is uniformly sampled from $\left\{ \boldsymbol{\lambda}_k \right\} _{k=0}^{T-1}$.  
\end{remark}

%%%%%%%%%%%%%%%%%%%%%%%%%%%%%%%%%%%%%%%%%%%%%%%%%%%%%%%%%%%%%%%%%%%%%%%%%%%%%%%%
\section{Numerical Experiments}
\label{sec:numerical_experiments}
In this section, we evaluate the performance of the proposed Bi-ZOL algorithm to solve the incentive-based load curtailment problem on a simulated network. 

\subsection{Simulation Setup}
\label{subsec:sim_setup}
The simulations are conducted on a 3-node network populated with 8 end-users randomly connected to the nodes. 
Each end-user manages 4 distinct flexible devices capable of providing load reduction.

The nodal curtailment incentives are constrained within the range $\Lambda = [0, 5]^3$ to ensure economic feasibility.
The end-users are modeled as rational agents minimizing a quadratic cost function as defined in \eqref{eq:user_lower}.
To capture population heterogeneity, the discomfort sensitivity coefficient $\alpha_{u,k}$ for each device is sampled uniformly from the interval $[0.5, 3.0]$.
Device capacities $C_{u,k}$ are randomly assigned to reflect diverse appliance ratings.

We implement the Bi-ZOL algorithm with a constant step size $\gamma = 10^{-3}$ and a smoothing radius $\delta = 10^{-3}$, values determined via empirical tuning.
To establish a ground truth for optimality, we leverage the structural properties of the lower-level problem.
Since the device response function \eqref{eq:clip_closed_form} is piecewise affine, the bilevel optimization problem can be exactly reformulated as a Mixed-Integer Linear Program (MILP) by introducing binary variables to model the saturation states of each device.
The global optimal solution $\boldsymbol{\lambda}^*$ and corresponding cost $\tilde{\varphi}(\boldsymbol{\lambda}^*)$ is obtained by solving this MILP using the Gurobi optimizer.

\subsection{Simulation Results}
\label{subsec:sim_results}
The convergence behavior of the algorithm and comparison with the global optimum is illustrated in Fig. \ref{fig:convergence_error}.
Fig. \ref{fig:convergence_error} shows the evolution of the objective function value. 
Bi-ZOL exhibits a steady descent and stabilizes in the neighborhood of the local optimum, validating the theoretical convergence guarantee to the $(\delta, \epsilon)$-FWSP.
Despite treating the user response as a black box, Bi-ZOL achieves a final objective value of $40.51$, which is remarkably close to the global optimum of $38.34$.
The relative optimality gap is approximately $5.66\%$, demonstrating that the proposed zeroth-order estimator effectively captures the descent direction even in the presence of nonsmooth response functions.

\begin{figure}[htb]
    \centering
    \includegraphics[width=\columnwidth]{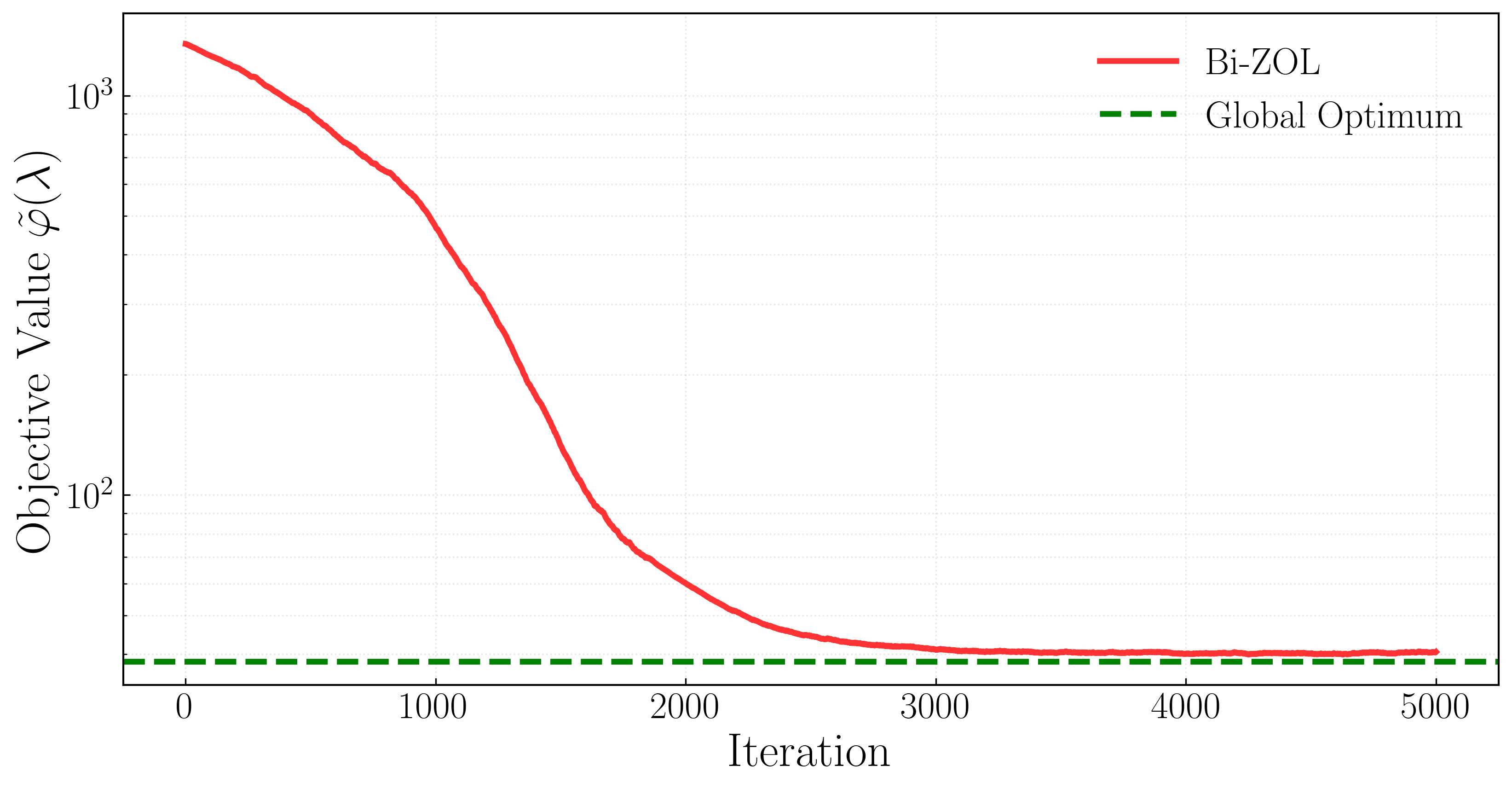}
    \caption{Convergence performance of Bi-ZOL. The objective value approaches the global optimum .}
    \label{fig:convergence_error}
\end{figure}

We further analyze the incentives and the corresponding end-users' load reduction responses.
Fig. \ref{fig:price_trajectory} illustrates the convergence trajectories of the incentives $\boldsymbol{\lambda}$ across three nodes. 
Throughout the iterations, the algorithm encounters nonsmooth points a total of 614 times, as indicated by the markers on each component's curve. The stable convergence, despite these frequent nonsmooth encounters, demonstrates that the smoothing technique within our Bi-ZOL algorithm effectively addresses the nonsmooth issues.

In Bi-ZOL iterations, $\lambda_1$ and $\lambda_3$ almost track the optimal prices, while $\lambda_2$ is a bit higher.
The reason is twofold.
First, we actually deal with a nonsmooth nonconvex optimization, in which global optimality is not guaranteed.
Our method only guarantees convergence to a $(\delta, \epsilon)$-FWSP, which is probably close to a local optimum.
Second, even though we select a most exact approximation for the hypergradient, the approximation error for the smoothing is still non-negligible.

\begin{figure}[htb]
    \centering
    \includegraphics[width=\columnwidth]{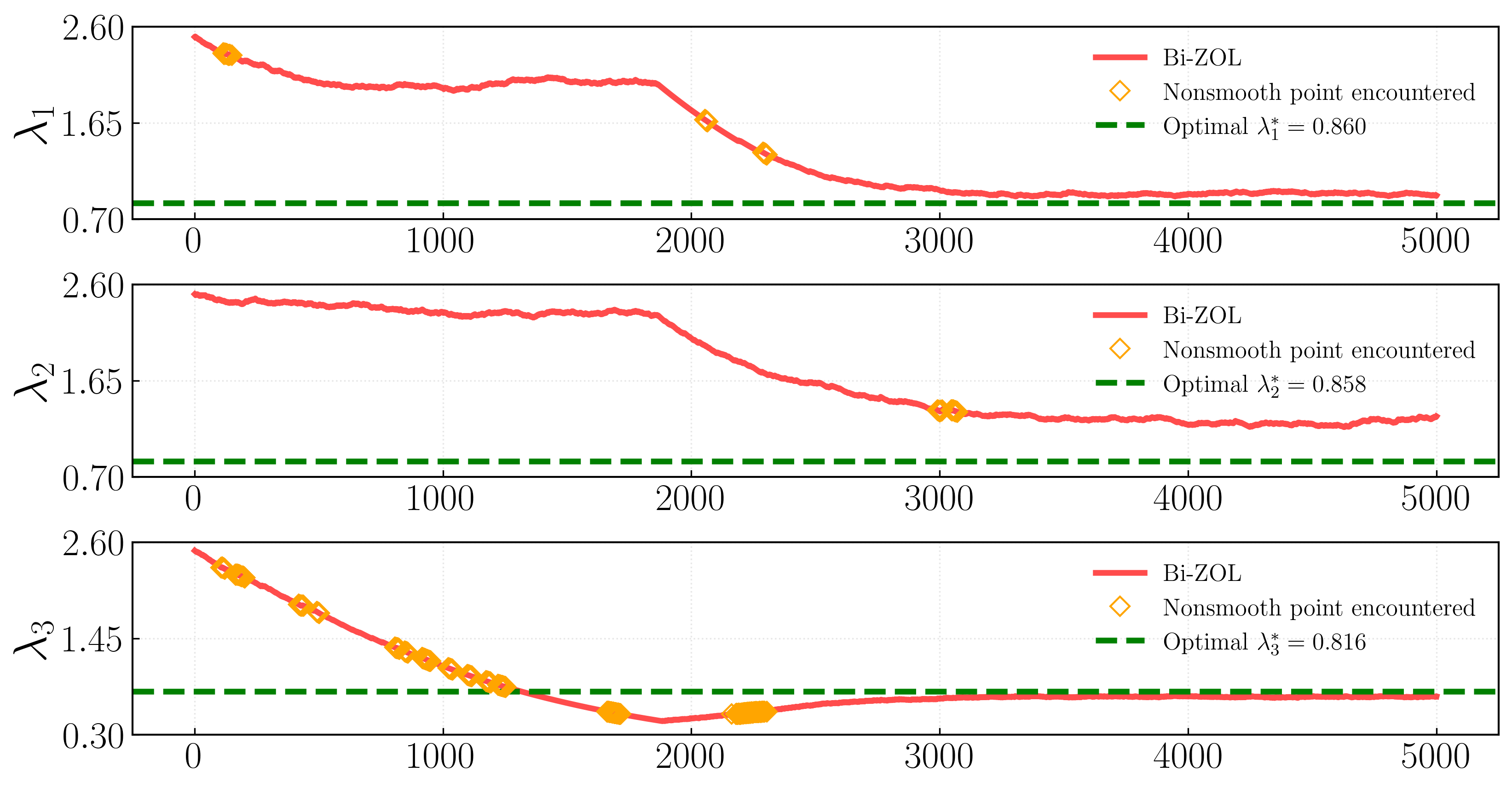}
    \caption{Incentive Trajectories. The yellow diamonds mark instances where the algorithm encounters nonsmooth points.
    Two nodal incentives among three nodes converge close to the optimal solution. 
    On the other hand, the incentive of the second node is a bit higher, which is probably due to the learning error and the local optimum.}
    \label{fig:price_trajectory}
\end{figure}
Fig. \ref{fig:response_trajectory} illustrates the load reduction response of the end-users.
The responses vary significantly across users, confirming the heterogeneity of the population.
Notably, at the converged state, four devices belonging to end-users 1, 3, and 5 reach their load reduction limits. This implies that the algorithm ultimately converges to a nonsmooth point.
Crucially, all end-users' responses are stable after iterations, which means our method is able to learn the end-users' responses and provide a stable incentive signal to control the loads.

\begin{figure}[htb]
    \centering
    \includegraphics[width=\columnwidth]{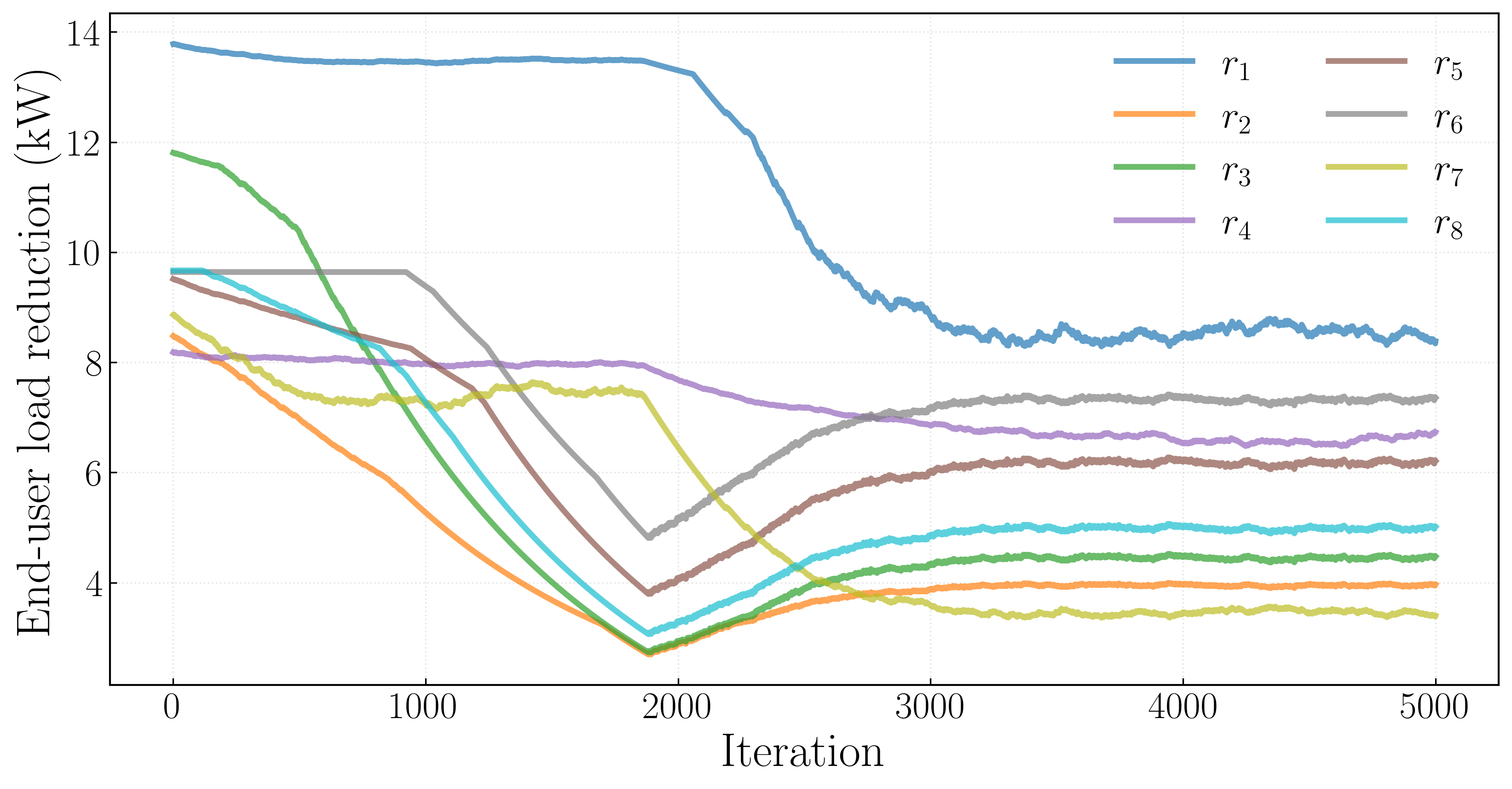}
    \caption{End-user Response Trajectories. The response of heterogeneous users stabilizes after iterations.}
    \label{fig:response_trajectory}
\end{figure}

\section{Conclusion}
\label{sec:conclusion}
This paper presents a bilevel optimization framework for incentive-based load curtailment.
By explicitly modeling the decision-making process of end-users subject to device capacity constraints, we identified that standard gradient-based methods fail due to undefined sensitivities at saturation points, while standard zeroth-order methods suffer from inefficiency by ignoring known system structures.
To bridge this gap, we proposed the Bi-ZOL algorithm, which synergizes analytical model information with learning response sensitivities.
Utilizing a structural hypergradient decomposition and a two-point estimator for sensitivity, Bi-ZOL rigorously handles the nonsmooth points in aggregate responses.
Theoretical analysis confirms that Bi-ZOL converges to a $(\delta, \epsilon)$-Frank-Wolfe stationary point under mild assumptions.
Furthermore, numerical experiments validate that our approach stabilizes the incentive near the global optimum.
Future work will extend this framework to incorporate grid constraints, exploring how voltage limits affect the learning of response sensitivities.
Further research could also investigate methods that reduce active probing, for instance by employing one-point gradient estimators or leveraging historical feedback to estimate sensitivities.
Additionally, investigating the tracking performance of Bi-ZOL in time-varying environments remains a promising direction for enhancing real-time grid operations.

% \section{AI Usage Disclosure}

\appendix 
\label{app:convergence}
% \subsection{Proof of Lemma \ref{lemma:jac_estimator}}
% \label{sec:appendix}
In this appendix, we provide the detailed proofs for the convergence of the Bi-ZOL algorithm. We use $\parallel\cdot\parallel$ to denote the $l_2$ norm for vectors and the spectral norm for matrices.

\subsection{Properties of the Zeroth-Order Estimator}
We begin by characterizing the statistical properties of the zeroth-order Jacobian estimator defined in \eqref{eq:jac_estimator}.

\begin{lemma}
    \label{lemma:jac_estimator}
    The two-point estimator \eqref{eq:jac_estimator} is an unbiased estimator of the smoothed Jacobian and satisfies the following variance bound:
    \begin{equation}
        \label{eq:jac_estimator_property}
        \begin{aligned}
        \mathbb{E}\!\left[\widehat{\mathbf{J}}{\mathbf{R}}\mid \boldsymbol{\lambda}\right]
        &= \mathbf{J}\mathbf{R}_{\delta}(\boldsymbol{\lambda}), \\
        \mathbb{E}\!\left[ \big\|\widehat{\mathbf{J}}{\mathbf{R}} \big\|^{2}\mid \boldsymbol{\lambda}\right]
        &\le 16\sqrt{2\pi}\,N^2 L_{R}^{2}.
        \end{aligned}
    \end{equation}
    \end{lemma}
    \begin{proof}
    For each $i\in\{1,\dots,N\}$, the $i$-th row of the Jacobian matrix $\widehat{\mathbf{J}}{\mathbf{R}}$ corresponds to the transposed vector estimator for the gradient of the $i$-th component function $R_i$:
    \begin{equation*}
        \label{eq:JR_i}
        \widehat{\mathbf{J}}{\mathbf{R}_i}
        \;\coloneqq\;
        \frac{N}{2\delta}\Big(R_i(\boldsymbol{\lambda}+\delta \mathbf{w})-R_i(\boldsymbol{\lambda}-\delta \mathbf{w})\Big)\mathbf{w}^\top
        \in\mathbb{R}^{1 \times N}.
    \end{equation*}        
    Consider the randomized smoothing of the scalar function $R_i$ over the unit ball:
    \begin{equation*}
        R_{i,\delta}(\boldsymbol{\lambda})
        \;\coloneqq\;
        \mathbb{E}_{\mathbf u\sim \mathrm{Unif}(\mathbb{B})}\big[R_i(\boldsymbol{\lambda}+\delta \mathbf u)\big].
    \end{equation*}
    From standard properties of zeroth-order estimators on the unit sphere in \cite{Lin2022}, the expectation of the column estimator is the gradient. Taking the transpose, we have:
    \begin{equation*}
        \mathbb{E}\big[\widehat{\mathbf{J}}{\mathbf{R}_i} \mid \boldsymbol{\lambda}\big]
        = \nabla R_{i,\delta}(\boldsymbol{\lambda})^\top.
    \end{equation*}
    Let $\mathbf R_\delta(\boldsymbol{\lambda})\coloneqq (R_{1,\delta}(\boldsymbol{\lambda}),\dots,R_{N,\delta}(\boldsymbol{\lambda}))^\top$. Stacking the expectations of the rows yields the Jacobian matrix:
    \begin{equation*}
        \mathbb{E}\big[\widehat{\mathbf{J}}{\mathbf{R}}\big]
        =
        \mathbf{J}\mathbf{R}_\delta(\boldsymbol{\lambda}).
    \end{equation*}
        
        To bound the second moment, we utilize the Frobenius norm decomposition:
        \begin{equation*}
        \big\|\widehat{\mathbf{J}}{\mathbf{R}}\big\|^{2}
        \leq \big\|\widehat{\mathbf{J}}{\mathbf{R}}\big\|_F^2
        =
        \sum_{i=1}^N \big\|\widehat{\mathbf{J}}{\mathbf{R}_i}\big\|^2.
        \end{equation*}
        Applying the bound for two-point estimators on Lipschitz functions from \cite[Lemma~D.1]{Lin2022}, we have for each $i$:
        \begin{equation*}
        \mathbb{E}\big[\|\widehat{\mathbf{J}}{\mathbf{R}_i}\|^2\mid \boldsymbol{\lambda}\big]
        \le 16\sqrt{2\pi}\,N\,L_i^2,
        \end{equation*}
        where $L_i$ is the Lipschitz constant of $R_i(\cdot)$. Summing over all $N$ rows yields:
        \begin{equation*}
        \mathbb{E}\big[\|\widehat{\mathbf{J}}{\mathbf{R}}\|^2\mid \boldsymbol{\lambda}\big]
        \le 16\sqrt{2\pi}\,N\sum_{i=1}^N L_i^2.
        \end{equation*}
        Using Lemma \ref{lemma:R_vector_Lipschitz}, we substitute $L_i \le L_R$ for all $i$, which implies $\sum_{i=1}^N L_i^2 \le N L_R^2$. Thus:
        \begin{equation*}
         \mathbb{E}\big[\|\widehat{\mathbf{J}}{\mathbf{R}}\|^2\mid \boldsymbol{\lambda}\big] \le 16\sqrt{2\pi}\,N^2 L_R^2,
        \end{equation*}
        which completes the proof.
    \end{proof}

\subsection{Smoothness Properties}
Next, we establish the Lipschitz continuity of the smoothed response and its Jacobian.

\begin{lemma}
    \label{lemma:R_delta_Lipschitz}
     $\mathbf{R}_{\delta}(\boldsymbol{\lambda})$ is Lipschitz continuous with Lipschitz constant $L_R$ and bounded by $R_{\max}$.
\end{lemma}
\begin{proof}
    By definition, we have:
    \begin{align*}
        \|\mathbf{R}_{\delta}(\boldsymbol{\lambda})-\mathbf{R}_{\delta}(\boldsymbol{\lambda}')\|
        & =\left\| \mathbb{E} \left[ \mathbf{R}(\boldsymbol{\lambda }+\delta \mathbf{u})-\mathbf{R}(\boldsymbol{\lambda }' +\delta \mathbf{u}) \right] \right\| 
        \\
        & \le \mathbb{E} \left[ \left\| \mathbf{R}(\boldsymbol{\lambda }+\delta \mathbf{u})-\mathbf{R}(\boldsymbol{\lambda }' +\delta \mathbf{u}) \right\| \right] 
        \\
        & \le L_R\left\| \boldsymbol{\lambda }-\boldsymbol{\lambda }' \right\|.
    \end{align*}
    The first inequality is due to $\left\| \mathbb{E} \left[ X \right] \right\| \le \mathbb{E} \left[ \left\| X \right\| \right]$,
    and the second inequality is due to the Lemma~\ref{lemma:R_vector_Lipschitz}.
    Similarly, we can prove $\left\| \boldsymbol{R}_{\delta}\left( \boldsymbol{\lambda } \right) \right\| \le R_{\max}$.
\end{proof}

\begin{lemma}
\label{lemma:JR_delta_Lipschitz_tight}
The $\mathbf{J}\mathbf{R}_\delta(\boldsymbol{\lambda})$ is Lipschitz continuous with Lipschitz constant $\frac{c N L_R}{\delta}$, where $c$ is a constant.
\end{lemma}
\begin{proof}
Fix any $i\in\{1,\dots,N\}$ and define the scalar smoothing
$
R_{i,\delta}(\boldsymbol{\lambda})
\coloneqq
\mathbb{E}_{\mathbf{u}\sim\mathrm{Unif}(\mathbb{B})}\!\left[R_i(\boldsymbol{\lambda}+\delta\mathbf{u})\right].
$
Since $R_i(\cdot)$ is $L_i$-Lipschitz by Proposition~\ref{prop:response_function_property},
$R_{i,\delta}$ is differentiable and its gradient is Lipschitz with constant $\frac{cL_i\sqrt{N}}{\delta}$, i.e.,
\begin{equation}
\label{eq:Ri_grad_Lipschitz}
\left\|\nabla R_{i,\delta}(\boldsymbol{\lambda})-\nabla R_{i,\delta}(\boldsymbol{\lambda}')\right\|
\le
\frac{cL_i\sqrt{N}}{\delta}\left\|\boldsymbol{\lambda}-\boldsymbol{\lambda}'\right\|.
\end{equation}
The inequality is due to \cite[Proposition~2.2]{Lin2022}.
Recall that the $i$-th row of $\mathbf{J}\mathbf{R}_\delta(\boldsymbol{\lambda})$ equals $\nabla R_{i,\delta}(\boldsymbol{\lambda})^\top$.
Therefore, using the Frobenius norm,
\begin{align*}
\left\|\mathbf{J}\mathbf{R}_\delta(\boldsymbol{\lambda})-\mathbf{J}\mathbf{R}_\delta(\boldsymbol{\lambda}')\right\|_F^2
&=
\sum_{i=1}^N
\left\|
\nabla R_{i,\delta}(\boldsymbol{\lambda})-\nabla R_{i,\delta}(\boldsymbol{\lambda}')
\right\|^2
\\
&\le
\sum_{i=1}^N
\left(\frac{cL_i\sqrt{N}}{\delta}\left\|\boldsymbol{\lambda}-\boldsymbol{\lambda}'\right\|\right)^2
\\
&=
\left(\frac{c N L_R}{\delta}\right)^2
\left\|\boldsymbol{\lambda}-\boldsymbol{\lambda}'\right\|^2.
\end{align*}
Taking square roots gives
\[
\left\|\mathbf{J}\mathbf{R}_\delta(\boldsymbol{\lambda})-\mathbf{J}\mathbf{R}_\delta(\boldsymbol{\lambda}')\right\|_F
\le
\frac{c N L_R}{\delta}\left\|\boldsymbol{\lambda}-\boldsymbol{\lambda}'\right\|.
\]
Finally, since $\|\cdot\|\le \|\cdot\|_F$, we finish the proof.
\end{proof}

\subsection{Properties of the Ancillary Function}
To prove the convergence of the Bi-ZOL algorithm, we utilize the ancillary function defined as:
\begin{equation}
    \label{eq:mediate_potential_app}
    \bar{\varphi}_{\delta}(\boldsymbol{\lambda }) \coloneqq \varphi \left( \boldsymbol{\lambda },\mathbf{R}_{\delta} \right)
    =\boldsymbol{\lambda}^{\top}\mathbf{R}_{\delta}(\boldsymbol{\lambda})\;+\;\rho E_{\delta}\left( \boldsymbol{\lambda } \right) ^2,
\end{equation}
where $E_{\delta}\coloneqq \mathbf{1}^\top\!\big(\mathbf{P}^{\mathrm{b}}-\mathbf{R}_{\delta}\big)\;-\;P^{\mathrm{target}}$. We now prove $\nabla \bar{\varphi}_{\delta}(\boldsymbol{\lambda })$ is Lipschitz continuous.

\begin{lemma}
    \label{lemma:mediate_potential_grad_Lipschitz_tight}
    The gradient of the ancillary function $\nabla \bar{\varphi}_{\delta}(\cdot)$ is Lipschitz continuous on $\Lambda$ with constant:
    \begin{equation}
    \begin{split}
        L_{\nabla \bar{\varphi}_\delta} = & \, L_R (2 + 2\rho N L_R) \\
        & + \frac{c N L_R}{\delta} \left(\Lambda_{\max} + 2\rho \sqrt{N} (|c_0| + \sqrt{N} R_{\max})\right),
    \end{split}
    \end{equation}
    where $\Lambda_{\max}\coloneqq \sup_{\boldsymbol{\lambda}\in\Lambda}\|\boldsymbol{\lambda}\|$.
    \end{lemma}
    
    \begin{proof}
    We first derive the expression for $\nabla \bar{\varphi}_{\delta}(\boldsymbol{\lambda})$.
    Let $c_0\coloneqq \mathbf{1}^{\top}\mathbf{P}^{\mathrm{b}}-P^{\mathrm{target}}$.
    Recall that $E_\delta(\boldsymbol{\lambda})=c_0-\mathbf{1}^\top \mathbf{R}_\delta(\boldsymbol{\lambda})$. Thus, $\nabla E_\delta(\boldsymbol{\lambda}) = -\mathbf{J}\mathbf{R}_\delta(\boldsymbol{\lambda})^\top \mathbf{1}$.
    Applying the chain rule to $\bar{\varphi}_{\delta}$:
    \begin{align*}
       \nabla \bar{\varphi}_{\delta}(\boldsymbol{\lambda})
    &=
    \mathbf{R}_\delta(\boldsymbol{\lambda})
    +
    \mathbf{J}\mathbf{R}_\delta(\boldsymbol{\lambda})^\top \boldsymbol{\lambda}
    -2\rho E_\delta(\boldsymbol{\lambda})\,\mathbf{J}\mathbf{R}_\delta(\boldsymbol{\lambda})^\top \mathbf{1}
    \\
    &=
    \mathbf{R}_\delta(\boldsymbol{\lambda})
    +
    \mathbf{J}\mathbf{R}_\delta(\boldsymbol{\lambda})^\top \mathbf{v}(\boldsymbol{\lambda}), 
    \end{align*}
    where we define the auxiliary vector $\mathbf{v}(\boldsymbol{\lambda})\coloneqq \boldsymbol{\lambda}-2\rho E_\delta(\boldsymbol{\lambda})\mathbf{1}$.
    
    Fix arbitrary $\boldsymbol{\lambda},\boldsymbol{\lambda}'\in\Lambda$ and let $\Delta\boldsymbol{\lambda}\coloneqq \boldsymbol{\lambda}-\boldsymbol{\lambda}'$.
    Using the triangle inequality:
    \begin{align*}
    \left\|\nabla \bar{\varphi}_{\delta}(\boldsymbol{\lambda})-\nabla \bar{\varphi}_{\delta}(\boldsymbol{\lambda}')\right\| 
    &\le
    \left\|\mathbf{R}_\delta(\boldsymbol{\lambda})-\mathbf{R}_\delta(\boldsymbol{\lambda}')\right\|
    \\
    &\hspace{-1.0em}+
    \left\|\mathbf{J}\mathbf{R}_\delta(\boldsymbol{\lambda})^\top \mathbf{v}(\boldsymbol{\lambda})
    -\mathbf{J}\mathbf{R}_\delta(\boldsymbol{\lambda}')^\top \mathbf{v}(\boldsymbol{\lambda}')\right\|.
    \end{align*}
    For the second term, we add and subtract $\mathbf{J}\mathbf{R}_\delta(\boldsymbol{\lambda})^\top \mathbf{v}(\boldsymbol{\lambda}')$:
    \begin{align*}
    % (LHS)
    &\big\|\mathbf{J}\mathbf{R}_\delta(\boldsymbol{\lambda})^\top \mathbf{v}(\boldsymbol{\lambda})
    -\mathbf{J}\mathbf{R}_\delta(\boldsymbol{\lambda}')^\top \mathbf{v}(\boldsymbol{\lambda}')\big\|
    \\
    % s1
    &\le
    \big\|\mathbf{J}\mathbf{R}_\delta(\boldsymbol{\lambda})^\top\big(\mathbf{v}(\boldsymbol{\lambda})-\mathbf{v}(\boldsymbol{\lambda}')\big)\big\|
    \\
    &\quad + 
    \big\|\big(\mathbf{J}\mathbf{R}_\delta(\boldsymbol{\lambda})-\mathbf{J}\mathbf{R}_\delta(\boldsymbol{\lambda}')\big)^\top \mathbf{v}(\boldsymbol{\lambda}')\big\|
    \\
    % s2
    &\le
    \big\|\mathbf{J}\mathbf{R}_\delta(\boldsymbol{\lambda})\big\|\cdot
    \big\|\mathbf{v}(\boldsymbol{\lambda})-\mathbf{v}(\boldsymbol{\lambda}')\big\|
    \\
    &\quad + 
    \big\|\mathbf{J}\mathbf{R}_\delta(\boldsymbol{\lambda})-\mathbf{J}\mathbf{R}_\delta(\boldsymbol{\lambda}')\big\|\cdot
    \big\|\mathbf{v}(\boldsymbol{\lambda}')\big\|.
\end{align*}
    
    We now bound each term. First, for $\|\mathbf{v}(\boldsymbol{\lambda})-\mathbf{v}(\boldsymbol{\lambda}')\|$, substituting the definition of $\mathbf{v}$:
    \begin{align*}
    \left\|\mathbf{v}(\boldsymbol{\lambda})-\mathbf{v}(\boldsymbol{\lambda}')\right\|
    &\le
    \left\|\Delta\boldsymbol{\lambda}\right\|
    +2\rho\left|E_\delta(\boldsymbol{\lambda})-E_\delta(\boldsymbol{\lambda}')\right|\cdot \|\mathbf{1}\|.
    \end{align*}
    The mismatch difference is bounded by:
    \begin{align*}
        \left|E_\delta(\boldsymbol{\lambda})-E_\delta(\boldsymbol{\lambda}')\right|
    &=
    \left|\mathbf{1}^\top\big(\mathbf{R}_\delta(\boldsymbol{\lambda}')-\mathbf{R}_\delta(\boldsymbol{\lambda})\big)\right|
    \\
    &\le
    \|\mathbf{1}\|\cdot \left\|\mathbf{R}_\delta(\boldsymbol{\lambda})-\mathbf{R}_\delta(\boldsymbol{\lambda}')\right\|
    \\
    &\le
    \sqrt{N}\,L_R\left\|\Delta\boldsymbol{\lambda}\right\|.
    \end{align*}
    Substituting $\|\mathbf{1}\|=\sqrt{N}$ yields:
    \begin{equation}
    \label{eq:v_diff_bound_tight}
    \left\|\mathbf{v}(\boldsymbol{\lambda})-\mathbf{v}(\boldsymbol{\lambda}')\right\|
    \le
    \left(1+2\rho N L_R\right)\left\|\Delta\boldsymbol{\lambda}\right\|.
    \end{equation}
    
    Next, we bound $\|\mathbf{v}(\boldsymbol{\lambda}')\|$. Using Lemma~\ref{lemma:R_delta_Lipschitz}:
    \begin{align*}
        |E_\delta(\boldsymbol{\lambda}')|
    &=
    |c_0-\mathbf{1}^\top \mathbf{R}_\delta(\boldsymbol{\lambda}')|
    \\
    &\le
    |c_0|+\|\mathbf{1}\|\cdot \|\mathbf{R}_\delta(\boldsymbol{\lambda}')\|
    \\
    &\le
    |c_0|+\sqrt{N}R_{\max}.
    \end{align*}
    Therefore,
    \begin{equation}
    \label{eq:v_bound_tight}
    \left\|\mathbf{v}(\boldsymbol{\lambda}')\right\|
    \le
    \Lambda_{\max}+2\rho\sqrt{N}\big(|c_0|+\sqrt{N}R_{\max}\big).
    \end{equation}
    
    Finally, recall that $\|\mathbf{J}\mathbf{R}_\delta(\boldsymbol{\lambda})\|\le L_R$ (because of Lemma~\ref{lemma:R_delta_Lipschitz}, we have $\mathbf{R}_\delta$ is $L_R$-Lipschitz and differentiable). 
    Using Lemma~\ref{lemma:JR_delta_Lipschitz_tight}, the Jacobian difference is bounded by $\frac{c N L_R}{\delta}\left\|\Delta\boldsymbol{\lambda}\right\|$.
    Combining these into the triangle inequality:
    \begin{align*}
    &\| \nabla \bar{\varphi}_{\delta}(\boldsymbol{\lambda}) - \nabla \bar{\varphi}_{\delta}(\boldsymbol{\lambda}') \| \\
    &\le L_R \|\Delta\boldsymbol{\lambda}\| + L_R(1+2\rho N L_R) \|\Delta\boldsymbol{\lambda}\| \\
    &\quad + \frac{c N L_R}{\delta} \Big(\Lambda_{\max} + 2\rho\sqrt{N}\big(|c_0|+\sqrt{N}R_{\max}\big)\Big) \|\Delta\boldsymbol{\lambda}\| \\
    &= \bigg[ L_R (2 + 2\rho N L_R) \\
    &\quad + \frac{c N L_R}{\delta} \Big(\Lambda_{\max} + 2\rho\sqrt{N}\big(|c_0|+\sqrt{N}R_{\max}\big)\Big) \bigg] \|\boldsymbol{\lambda}-\boldsymbol{\lambda}'\|.
\end{align*}
    This completes the proof.
    \end{proof}

\subsection{Bias Error Bound}
With the smoothness of the ancillary function established, we now quantify the difference between approximate hypergradient $\nabla \tilde{\varphi}_{\delta}$ (which uses the observed response $\mathbf{R}$) and the gradient of the ancillary function $\nabla \bar{\varphi}_{\delta}$ (which uses the smoothed response $\mathbf{R}_\delta$).

\begin{lemma}
    \label{lemma:bias_error}
    The difference between the approximate hypergradient and the gradient of the ancillary function is bounded by:
    \begin{equation}
        \label{eq:bias_error}
        \|\nabla \bar{\varphi}_{\delta}(\boldsymbol{\lambda })-\nabla \tilde{\varphi}_{\delta}(\boldsymbol{\lambda })\| \le \left( 1+2\rho NL_R \right) \delta L_R.
    \end{equation}
\end{lemma}
\begin{proof}
    First, we show that the partial derivative $\nabla_2 \varphi (\boldsymbol{\lambda}, \cdot)$ is Lipschitz continuous with respect to the response vector $\mathbf{R}$.
    Recall from \eqref{eq:partials_closedform} that $\nabla_2 \varphi(\boldsymbol{\lambda}, \mathbf{R}) = \boldsymbol{\lambda} - 2\rho (\mathbf{1}^\top \mathbf{P}^{\mathrm{b}} - \mathbf{1}^\top \mathbf{R} - P^{\mathrm{target}}) \mathbf{1}$.
    For any vectors $\mathbf{R}, \mathbf{R}'$:
    \begin{align*}
        \left\| \nabla _2\varphi \left( \boldsymbol{\lambda },\mathbf{R}' \right) -\nabla _2\varphi \left( \boldsymbol{\lambda },\mathbf{R} \right) \right\| 
        &= \left\| -2\rho (-\mathbf{1}^\top \mathbf{R}' - (-\mathbf{1}^\top \mathbf{R})) \mathbf{1} \right\|
        \\
        &= 2\rho \left\| (\mathbf{1}^{\top}\mathbf{R}-\mathbf{1}^{\top}\mathbf{R}')\mathbf{1} \right\| 
        \\
        &\le 2\rho \|\mathbf{1}\| \cdot |\mathbf{1}^{\top}(\mathbf{R}-\mathbf{R}')|
        \\
        &\le 2\rho \sqrt{N} \cdot \|\mathbf{1}\| \|\mathbf{R}-\mathbf{R}'\|
        \\
        &= 2\rho N \left\| \mathbf{R}-\mathbf{R}' \right\|.
    \end{align*}
    Thus, the Lipschitz constant is $L_{\varphi 2} \coloneqq 2\rho N$.

    Next, we expand the difference between the gradients. Recall $\nabla \tilde{\varphi}_{\delta} = \mathbf{R} + \mathbf{J}\mathbf{R}_\delta^\top \nabla_2\varphi(\boldsymbol{\lambda}, \mathbf{R})$ and $\nabla \bar{\varphi}_{\delta} = \mathbf{R}_\delta + \mathbf{J}\mathbf{R}_\delta^\top \nabla_2\varphi(\boldsymbol{\lambda}, \mathbf{R}_\delta)$.
    \begin{align*}
        &\left\| \nabla \bar{\varphi}_{\delta}(\boldsymbol{\lambda })-\nabla \tilde{\varphi}_{\delta}(\boldsymbol{\lambda }) \right\| 
        \\
        &=\left\| (\mathbf{R}_{\delta} -\mathbf{R}) + \mathbf{J}\mathbf{R}_{\delta}^{\top}\left( \nabla _2\varphi \left( \boldsymbol{\lambda },\mathbf{R}_{\delta} \right) -\nabla _2\varphi \left( \boldsymbol{\lambda },\mathbf{R} \right) \right) \right\| 
        \\
        &\le \left\| \mathbf{R}_{\delta} -\mathbf{R} \right\| + \|\mathbf{J}\mathbf{R}_{\delta}\| \left\| \nabla _2\varphi \left( \boldsymbol{\lambda },\mathbf{R}_{\delta} \right) -\nabla _2\varphi \left( \boldsymbol{\lambda },\mathbf{R} \right) \right\| 
        \\
        &\le \left\| \mathbf{R}_{\delta} -\mathbf{R} \right\| + L_R (2\rho N) \left\| \mathbf{R}_{\delta} -\mathbf{R} \right\| 
        \\
        &= (1 + 2\rho N L_R) \left\| \mathbf{R}_{\delta} -\mathbf{R} \right\|.
    \end{align*}
    The second inequality uses $\|\mathbf{J}\mathbf{R}_\delta(\boldsymbol{\lambda})\|\le L_R$ (from Lemma~\ref{lemma:R_delta_Lipschitz}).

    Finally, we bound the smoothing error $\|\mathbf{R}_{\delta}(\boldsymbol{\lambda}) - \mathbf{R}(\boldsymbol{\lambda})\|$. Using the definition of randomized smoothing over the unit ball $\mathbb{B}$:
    \begin{align*}
        \|\mathbf{R}_{\delta}(\boldsymbol{\lambda}) - \mathbf{R}(\boldsymbol{\lambda})\| 
        &= \left\| \mathbb{E}_{\mathbf{u}\sim \mathrm{Unif}(\mathbb{B})}[\mathbf{R}(\boldsymbol{\lambda} + \delta \mathbf{u})] - \mathbf{R}(\boldsymbol{\lambda}) \right\|
        \\
        &\le \mathbb{E}_{\mathbf{u}}[\|\mathbf{R}(\boldsymbol{\lambda} + \delta \mathbf{u}) - \mathbf{R}(\boldsymbol{\lambda})\|]
        \\
        &\le \mathbb{E}_{\mathbf{u}}[L_R \|\delta \mathbf{u}\|] 
        \\
        &= \delta L_R \mathbb{E}_{\mathbf{u}}[\|\mathbf{u}\|] 
        \\
        &\le \delta L_R.
    \end{align*}

    The first inequality is due to the Jensen's inequality.
    Substituting this back yields the final bound:
    \begin{equation*}
        \left\| \nabla \bar{\varphi}_{\delta}(\boldsymbol{\lambda })-\nabla \tilde{\varphi}_{\delta}(\boldsymbol{\lambda }) \right\| \le \left( 1+2\rho NL_R \right) \delta L_R.
    \end{equation*}
\end{proof}

\subsection{Proof of Theorem \ref{thm:convergence}}
\begin{proof}
    Let $\hat{\mathbf{z}}_k \coloneqq \arg\max_{\mathbf{z}\in \Lambda} \langle \mathbf{z}, -\nabla \tilde{\varphi}_{\delta}(\boldsymbol{\lambda}_k) \rangle$ be the optimal direction for the approximate hypergradient, and recall that $\mathbf{z}_k$ is the oracle output for the stochastic estimate $\widehat{\mathbf{g}}_k$.
    
    Using the Lipschitz smoothness of the ancillary function $\bar{\varphi}_{\delta}$ (Lemma~\ref{lemma:mediate_potential_grad_Lipschitz_tight}) and the update $\boldsymbol{\lambda}_{k+1} = \boldsymbol{\lambda}_k + \gamma(\mathbf{z}_k - \boldsymbol{\lambda}_k)$:
    \begin{align*}
        \bar{\varphi}_{\delta}(\boldsymbol{\lambda}_{k+1}) 
        &\le \bar{\varphi}_{\delta}(\boldsymbol{\lambda}_k) + \gamma \langle \nabla \bar{\varphi}_{\delta}(\boldsymbol{\lambda}_k), \mathbf{z}_k-\boldsymbol{\lambda}_k \rangle 
        \\
        & \qquad + \frac{L_{\nabla \bar{\varphi}_{\delta}}\gamma^2}{2} \|\mathbf{z}_k-\boldsymbol{\lambda}_k\|^2
        \\
        &\le \bar{\varphi}_{\delta}(\boldsymbol{\lambda}_k) + \gamma \langle \nabla \bar{\varphi}_{\delta}(\boldsymbol{\lambda}_k), \mathbf{z}_k-\boldsymbol{\lambda}_k \rangle + \frac{L_{\nabla \bar{\varphi}_{\delta}} D^2}{2} \gamma^2.
    \end{align*}
    We decompose the inner product term to isolate the FW gap for $\nabla \tilde{\varphi}_{\delta}$:
    \begin{align*}
   & \langle \nabla \bar{\varphi}_{\delta}(\boldsymbol{\lambda}_k), \mathbf{z}_k-\boldsymbol{\lambda}_k \rangle \notag \\
    &= 
    \underbrace{\langle \nabla \tilde{\varphi}_{\delta}(\boldsymbol{\lambda}_k), \hat{\mathbf{z}}_k-\boldsymbol{\lambda}_k \rangle}_{\text{True Gap } (-\mathcal{G}_\delta)}
    + 
    \underbrace{\langle \nabla \tilde{\varphi}_{\delta}(\boldsymbol{\lambda}_k), \mathbf{z}_k-\hat{\mathbf{z}}_k \rangle}_{\text{Optimization Error}} \notag \\
    &\quad + 
    \underbrace{\langle \nabla \bar{\varphi}_{\delta}(\boldsymbol{\lambda}_k)-\nabla \tilde{\varphi}_{\delta}(\boldsymbol{\lambda}_k), \mathbf{z}_k-\boldsymbol{\lambda}_k \rangle}_{\text{Bias Error}}.
\end{align*}
    
    \textbf{Bounding the Optimization Error:}
    Recall that $\mathbf{z}_k$ minimizes $\langle \widehat{\mathbf{g}}_k, \cdot \rangle$. Thus, $\langle \widehat{\mathbf{g}}_k, \mathbf{z}_k-\boldsymbol{\lambda}_k \rangle \le \langle \widehat{\mathbf{g}}_k, \hat{\mathbf{z}}_k-\boldsymbol{\lambda}_k \rangle$. 
    We can write:
    \begin{align*}
        \langle \nabla \tilde{\varphi}_{\delta}(\boldsymbol{\lambda}_k), \mathbf{z}_k-\hat{\mathbf{z}}_k \rangle
        &= \langle \nabla \tilde{\varphi}_{\delta}(\boldsymbol{\lambda}_k) - \widehat{\mathbf{g}}_k, \mathbf{z}_k-\hat{\mathbf{z}}_k \rangle 
        \\
        &\qquad + \langle \widehat{\mathbf{g}}_k, \mathbf{z}_k-\hat{\mathbf{z}}_k \rangle
        \\
        &\le \langle \nabla \tilde{\varphi}_{\delta}(\boldsymbol{\lambda}_k) - \widehat{\mathbf{g}}_k, \mathbf{z}_k-\hat{\mathbf{z}}_k \rangle + 0
        \\
        &\le \|\nabla \tilde{\varphi}_{\delta}(\boldsymbol{\lambda}_k) - \widehat{\mathbf{g}}_k\| \cdot \|\mathbf{z}_k-\hat{\mathbf{z}}_k\|
        \\
        &\le \|\nabla \tilde{\varphi}_{\delta}(\boldsymbol{\lambda}_k) - \widehat{\mathbf{g}}_k\| D.
    \end{align*}

    \textbf{Bounding the Bias Error:}
    Using Lemma~\ref{lemma:bias_error}:
    \begin{align*}
        \langle \nabla \bar{\varphi}_{\delta}(\boldsymbol{\lambda}_k)-\nabla \tilde{\varphi}_{\delta}(\boldsymbol{\lambda}_k), \mathbf{z}_k-\boldsymbol{\lambda}_k \rangle
        &\le \|\nabla \bar{\varphi}_{\delta}(\boldsymbol{\lambda}_k)-\nabla \tilde{\varphi}_{\delta}(\boldsymbol{\lambda}_k)\| D
        \\
        &\le (1+2\rho N L_R)L_R \delta D 
        \\
        &= C_{\mathrm{bias}} D \delta.
    \end{align*}

    Substituting these back into the descent inequality and rearranging:
    \begin{align*}
        &\gamma \mathcal{G}_{\delta}(\boldsymbol{\lambda}_k) 
        \le 
        \bar{\varphi}_{\delta}(\boldsymbol{\lambda}_k) - \bar{\varphi}_{\delta}(\boldsymbol{\lambda}_{k+1})
        \\
        &+ \gamma \|\nabla \tilde{\varphi}_{\delta}(\boldsymbol{\lambda}_k) - \widehat{\mathbf{g}}_k\| D
        + \gamma C_{\mathrm{bias}} D \delta
        + \frac{L_{\nabla \bar{\varphi}_{\delta}} D^2}{2} \gamma^2.
    \end{align*}
    Taking the expectation and summing over $k=0, \dots, T-1$:
    \begin{align*}
        \frac{1}{T}\sum_{k=0}^{T-1} \mathbb{E}[\mathcal{G}_{\delta}(\boldsymbol{\lambda}_k)]
        &\le 
        \frac{\Delta_0}{\gamma T}
        +
        \frac{D}{T} \sum_{k=0}^{T-1} \mathbb{E}[\|\nabla \tilde{\varphi}_{\delta}(\boldsymbol{\lambda}_k) - \widehat{\mathbf{g}}_k\|]
        \\
        &+
        C_{\mathrm{bias}} D \delta
        +
        \frac{L_{\nabla \bar{\varphi}_{\delta}} D^2}{2} \gamma.
    \end{align*}
    
    \textbf{Bounding the Stochastic Gradient Error:}
    Recall $\nabla \tilde{\varphi}_{\delta} = \nabla_1 \varphi + \mathbf{J}\mathbf{R}_{\delta}^\top \nabla_2 \varphi$ and $\widehat{\mathbf{g}}_k = \nabla_1 \varphi + \widehat{\mathbf{J}}{\mathbf{R}}^\top \nabla_2 \varphi$.

    For the second term in the right-hand side, we have:
    \begin{align*}
        &\frac{1}{T}\sum_{k=0}^{T-1}{\mathbb{E} \left[ \left\| \nabla \tilde{\varphi}_{\delta}(\boldsymbol{\lambda }_k)-\widehat{\mathbf{g}}(\boldsymbol{\lambda }_k) \right\| \right]}
\\
&\overset{\left( s.1 \right)}{\le} \sqrt{\frac{1}{T}\sum_{k=0}^{T-1}{\mathbb{E} \left[ \left\| \nabla \tilde{\varphi}_{\delta}(\boldsymbol{\lambda }_k)-\widehat{\mathbf{g}}(\boldsymbol{\lambda }_k) \right\| ^2 \right]}}
\\
&\overset{\left( s.2 \right)}{\le} M_{\varphi 2}\sqrt{\frac{1}{T}\sum_{k=0}^{T-1}{\mathbb{E} \left[ \big\| \mathbf{J}\mathbf{R}_{\delta}-\widehat{\mathbf{J}}{\mathbf{R}} \big\|_F ^2 \right]}}
\\
&\overset{\left( s.3 \right)}{\le} M_{\varphi 2}\sqrt{\frac{1}{T}\sum_{k=0}^{T-1}{\mathbb{E} \left[ \big\| \widehat{\mathbf{J}}{\mathbf{R}} \big\|_F ^2 \right]}}
\\
&\overset{\left( s.4 \right)}{\le} M_{\varphi 2}\sqrt{16\sqrt{2\pi}\,N^2 L_{R}^{2}}
\\
&=4M_{\varphi 2}\left( 2\pi \right) ^{1/4}N L_R.
    \end{align*}
where $M_{\varphi 2}$ is the upper bound of $\nabla_2 \varphi$.
The $\left( s.1 \right)$ is due to the AM-QM inequality and Jensen’s inequality.
The $\left( s.2 \right)$ follows from the upper bound of $\nabla_2 \varphi$ and $\|\boldsymbol{A}^\top \boldsymbol{v}\| \le \|\boldsymbol{A}\|_F \|\boldsymbol{v}\|$.
The $\left( s.3 \right)$ follows from the fact $\mathbb{E} [\parallel X-E[X]\parallel_F ^2]\le \mathbb{E} [\parallel X\parallel_F ^2]$ by conditional unbiasedness in Lemma~\ref{lemma:jac_estimator}.
The $\left( s.4 \right)$ is from the proof in Lemma~\ref{lemma:jac_estimator}.
    Substituting this constant bound into the summation yields the final result.
\end{proof}

\bibliographystyle{ieeetr}
\bibliography{references}
\balance

\endgroup
\end{document}